\documentclass[conference]{IEEEtran}
\IEEEoverridecommandlockouts
\usepackage{cite}
\usepackage{amsmath,amssymb,amsfonts}
\usepackage{amsthm}
\usepackage{algorithmic}
\usepackage{graphicx}
\usepackage{textcomp}
\usepackage{xcolor}
\usepackage{subfig}
\usepackage[misc]{ifsym}

\def\BibTeX{{\rm B\kern-.05em{\sc i\kern-.025em b}\kern-.08em
    T\kern-.1667em\lower.7ex\hbox{E}\kern-.125emX}}

\newtheorem{lemma}{Lemma}

\begin{document}

\title{Online Scheduling of Transmission and Processing for AoI Minimization with Edge Computing}

\author{
    \IEEEauthorblockN{
        Jianhang Zhu, Jie Gong
    }

    \IEEEauthorblockA{
        \textit{Guangdong Key Laboratory of Information Security Technology},\\
        \textit{School of Computer Science and Engineering},
        \textit{Sun Yat-sen University},
        Guangzhou, 510006, China \\
        Email:
        zhujh26@mail2.sysu.edu.cn,
        gongj26@mail.sysu.edu.cn
    }
}

\maketitle

\begin{abstract}
Age of Information (AoI), which measures the time elapsed since the generation of the last received packet at the destination, is a new metric for real-time status update tracking applications. In this paper, we consider a status-update system in which a source node samples updates and sends them to an edge server over a delay channel. The received updates are processed by the server with an infinite buffer and then delivered to a destination. The channel can send only one update at a time, and the server can process one at a time as well. The source node applies \emph{generate-at-will} model according to the state of the channel, the edge server, and the buffer. We aim to minimize the average AoI with \emph{independent and identically distributed} transmission time and processing time. We consider three online scheduling policies. The first one is the optimal \emph{long wait} policy, under which the source node only transmits a new packet after the old one is delivered. Secondly, we propose a \emph{peak age threshold} policy, under which the source node determines the sending time based on the estimated peak age of information (PAoI). Finally, we improve the \emph{peak age threshold} policy by considering a postponed plan to reduce the waiting time in the buffer. The AoI performance under these policies is illustrated by numerical results with different parameters.

\end{abstract}


\section{Introduction}

With the increase of smart devices, the Internet-of-Things (IoT) has emerged as a new digital structure, which connects billions of things, such as small sensors, wearables, vehicles, and actuators, to the Internet. Thanks to the vast data collection capability of IoT, many real-time applications such as remote monitoring and control, phase data update in the smart grid, environment monitoring for autonomous driving, etc., can be realized, where the information freshness directly affects application performance. The concept of \emph{age of information} (AoI), or simply \emph{age}, defined as the time elapsed since the generation of the latest received update, was firstly proposed in \cite{kaul2012real} as a new metric to quantify information freshness. On the other hand, information in packets generated by a device needs to be extracted via computation. Due to the limited computing resources of the end device, using the edge server in the \emph{mobile edge computing} (MEC) \cite{MEC} to process packets is an effective method. Thus, it is a critical issue to jointly optimize transmission and processing for a computational-intensive status update.

We consider the status update system with MEC as shown in Fig. \ref{fig:system}. Since both transmission and processing can be regarded as services, a MEC system can be regarded as a two-hop network, where the first hop is transmission and the second hop is processing. There are a lot of research efforts on two-hop networks. In particular, the age minimization in a two-hop relay system, under a resource constraint on the average number of forwarding operations at the relay, is considered in \cite{TwohopCMDP}. The result shows that the optimal policy is multi-threshold. A two-hop network with energy harvesting nodes is studied in \cite{Twohop}, where both transmission and relay consume energy. There are also many papers that analyze AoI in multi-hop networks. Ref. \cite{Sunmultihop} studies the age of a single information flow in an  interference-free multi-hop network. A preemptive last-generated, first-served (LGFS) policy is proven to be optimal. The age of multi-hop multicast networks is examined in \cite{BBmultihop}.

\begin{figure}[!t]
    \centering
    \includegraphics[width=0.7\linewidth]{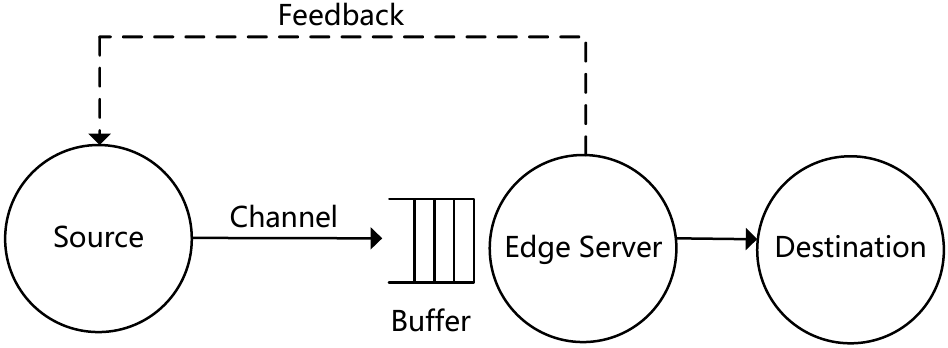}
    \caption{System model.}\label{fig:system}
\end{figure}

In recent years, there are rising research interests on the age of information in the MEC system. The average AoI with exponential transmission time and service time was analyzed in \cite{US1}. The offline scheduling problem is studied in \cite{US2}, where transmission and computing are jointly scheduled to minimize AoI. Task scheduling, computational offloading, and energy consumption are considered together to minimize AoI in \cite{Xian2019}. Ref. \cite{Pzhou2021} considers the collaboration of an edge server and a cloud center. To reduce AoI in the cloud center, the computing offloading policy in the edge server is investigated. In order to fully utilize edge resources and minimize the overall AoI of the whole MEC system, an online auction mechanism, called PreDisc, is proposed in \cite{Hlv2021}. Since the task processing in the MEC system usually includes multiple stages, joint scheduling to optimize AoI is a difficult problem. A most closely related work \cite{Pzhou2021Opt} studied the average AoI under different packet management rules in a tandem queue, where the packet arrival is a Poisson process. Because the packet processing in the edge server is influenced by the packet generation process, calculating the optimal scheduling policy in a MEC system is typically difficult. Most studies consider random packet generation or slotted systems to simplify the analysis. In this paper, the source node applies \emph{generate-at-will} model in continuous time.

A simple policy for the generate-at-will model is that the source node generates a new packet only after the old packet is delivered to the destination, termed as \emph{long wait} policy. The simplified system can be viewed as a single hop transmission as studied in \cite{SunUp}. Under this simplified strategy, packets will never enter the buffer of the edge server to wait. Intuitively, the long wait policy helps reduce the age of information. However, it is far from optimal in many cases. The following example can explain this.

\textit{Example: Suppose the source node sends a sequence of packets to the edge server. The transmission time of each packet is 2 and the processing time at the edge server is 1. Suppose that packet 1 is transmitted at time 0, and the initial AoI is 3. Under the long wait policy, the average AoI is 4.5 if the source node immediately sends a new packet when the old packet is delivered. For comparison, consider the source continuously transmitting packets and the edge server processing packets whenever available. Thus, the next packet starts to transmit when the previous one is being processed at the edge server. The average AoI is 4 under this policy. It is obvious that the long wait policy is not optimal.}

The above example shows that the source node sending packet when the edge server is busy can effectively reduce AoI. However, it may also cause the packet to wait in the buffer and thus increase AoI. Our goal is to optimize the source node's packets sending policy to minimize the average AoI of the destination. In order to minimize AoI, the arrived packet at the edge server or in the buffer should be served immediately when the server is idle.

In this paper, we study and compare three scheduling policies. Firstly, the optimal long wait policy is given, which has a threshold structure. Secondly, the source node is allowed to send a new packet when the edge server is busy. The source node estimates the \emph{peak age of information} (PAoI) and decides the transmission time according to a constant threshold. This policy is called the \emph{peak age threshold} policy. Finally, we improve the peak age threshold policy by trying to prevent packets from entering the buffer: the source node compares the estimated arrival time of a new packet at the server versus the delivery time of the old packet in the server and only sends the new packet when the former is larger. Numerical results illustrate the different behaviors of the AoI curve with different parameters.


\section{System model and problem formulation}
\label{sec:system}
\subsection{System Model}
We consider an information update system depicted in Fig.~\ref{fig:system}, where a source node generates update packets and sends them to an edge server through a channel and then the server processes the packets and sends the result to a destination. The source node generates and submits update packets at successive times $t_0,t_1,\cdots$ and the server node starts to serve the packets at successive times $c_0,c_1,\cdots$. The server will send one-bit feedback to the source when a packet arrives at the server (e.g., send 1 back to the source) or is delivered to the destination (e.g., send 0 back). When the edge server is idle, the oldest packet in the buffer should be processed immediately, otherwise, it will cause an unnecessary waiting time. Therefore, the source node can access the idle/busy state of the channel and the server, and the number of packets in the buffer through feedback. After the source sends update $k-1$ , the channel will be busy until it receives the $k$-th 1, and the server will be busy until it receives the $k$-th 0. Suppose the source has received 0 from packets $1,2,\cdots,l$, there exist $k-l-1$ packets in the buffer.

\begin{figure}[!t]
    \centering
    \includegraphics[width=0.7\linewidth]{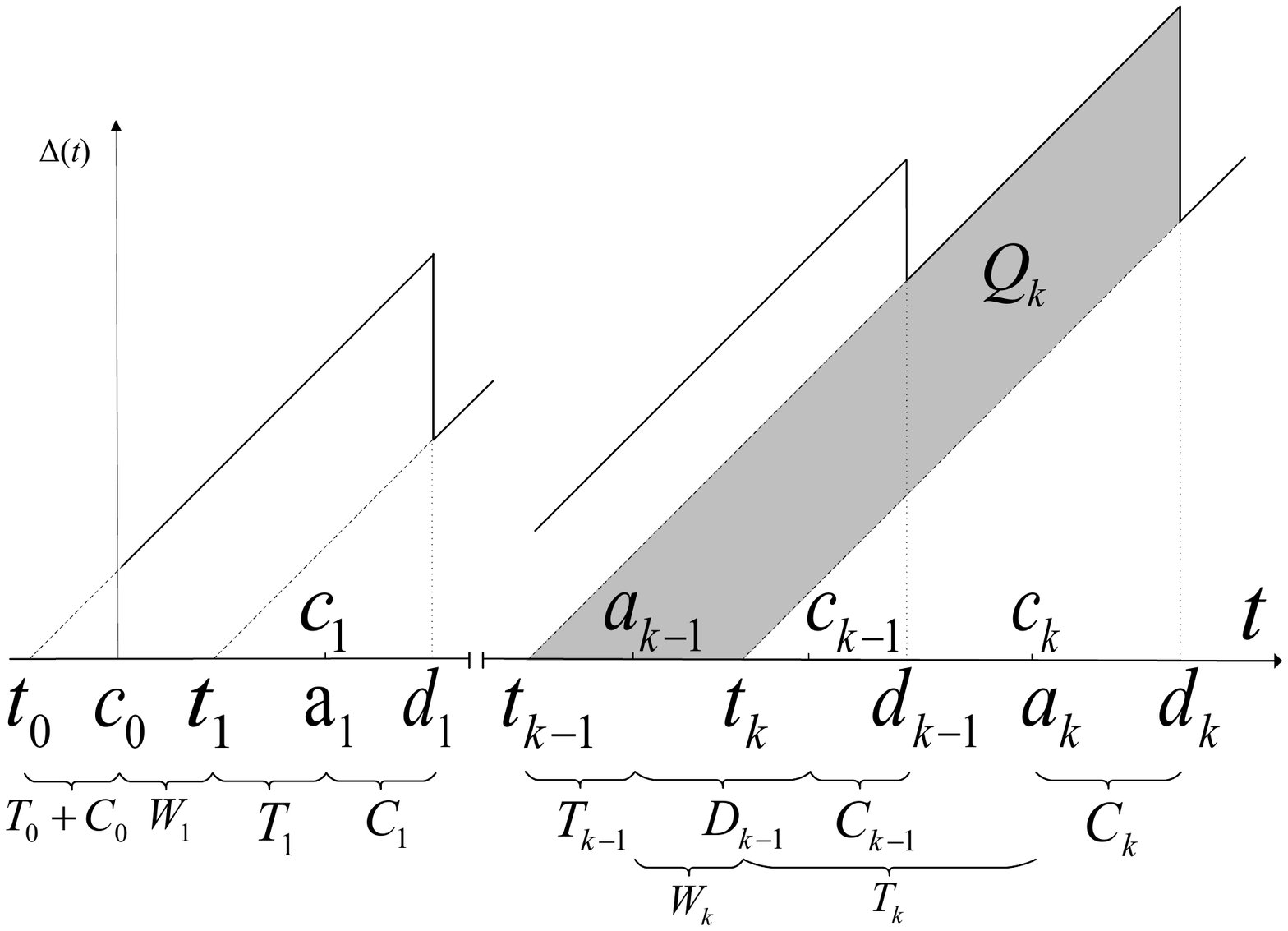}
    \caption{Age curve.}\label{fig:age}
\end{figure}

Suppose that packet $k$ is generated at time $t_k$ with transmission time duration $T_k\ge 0$, it will arrive at the server at time $a_k=t_k+T_k$. Suppose that packet $k$ is processed by the server at time $c_k$, it will be delivered to the destination at time $d_k=c_k+C_k$, where $C_k$ is the processing time duration of packet $k$. Since packet $k$ should be processed immediately after it arrives at the server and packets $k-1$ is delivered, we have
\begin{equation}
    c_k=\max\{t_k+T_k,c_{k-1}+C_{k-1}\}.\label{eq:ck}
\end{equation}
We assume that packet 0 is submitted to an idle channel at time $t_0=-T_0-C_0$ and is delivered at time $d_0=c_0+C_0=0$ with $c_0=-C_0$, as shown in Fig. \ref{fig:age}. When packet $k-1$ arrives at the server, the source will wait for a time duration $W_k>0$ and then send packet $k$ at time $t_{k-1}+T_{k-1}+W_k$, that is,
\begin{equation}
    t_{k}=t_{k-1}+T_{k-1}+W_k.\label{def:wk}
\end{equation}
Denote $D_k$ as the waiting time of packet $k$ in the buffer, we have
\begin{equation}
    D_k=c_k-t_k-T_k.\label{def:dk}
\end{equation}
We assume that $T_1,T_2,\cdots$ are \emph{independent and identically distributed} (i.i.d.) with mean $\mathbb{E}[T]$ and $C_1,C_2,\cdots$ are \emph{i.i.d.} with mean $\mathbb{E}[C]$, where $0< \mathbb{E}[T],\mathbb{E}[C]<\infty$.

At any time $t$, the most recently received packet is generated at time
\begin{equation}
    U(t)=\text{max}\{t_k|c_k+C_k\le t\}.
\end{equation}
The age of information $\Delta(t)$ is defined as
\begin{equation}
    \Delta(t)=t-U(t),
\end{equation}
As shown in Fig. \ref{fig:age}, the age $\Delta(t)$ jumps down when a new packet is delivered to the destination, and increases linearly otherwise.

\subsection{Problem Formulation}
We consider the long-term average age minimization problem. We divide the area under the age curve $\Delta(t)$ into a sum of multiple disjoint trapezoids $Q_k$'s, as shown in Fig. \ref{fig:age}. We have
\begin{equation}
    Q_k=\frac{1}{2}\left(c_k+C_k-t_{k-1}\right)^2-\frac{1}{2}\left(c_k+C_k-t_{k}\right)^2.\label{def:qk}
\end{equation}
In the interval $[t_0,t_n]$, where $t_n=t_0+\sum_{k=1}^{n}(T_{k-1}+W_k)$ and $n\rightarrow \infty$, the average AoI can be expressed as
\begin{equation}
    \overline{\Delta}=\underset{n \rightarrow \infty}{\text{lim}}\frac{\mathbb{E}[\sum_{k=1}^{n}Q_k]}{\mathbb{E}[\sum_{k=1}^{n}(T_{k-1}+W_k)]}.\label{eq:longAage}
\end{equation}
We can minimize the average age \eqref{eq:longAage} by controlling the sequence of waiting times $(W_1,W_2,...)$. Denote $\pi=(W_1,W_2,...)$ as a packet generation policy. We consider causal policies, in which $W_k$ is determined based on historical information and current system status, i.e., the idle/busy state of the server and the number of packets in the buffer, as well as the distribution of the transmission and processing times $(T_0,C_0,T_1,C_1,\cdots)$. Note that to determine $W_k$, $D_{k-1}$ and $C_{k-1}$ may be historical information or future information, because packet $k-1$ may still be processed or wait in the buffer when packet $k$ is generated. Specifically, $D_{k-1}$ is available when $t\ge c_{k-1}$ and $C_{k-1}$ is available when $t\ge c_{k-1}+C_{k-1}$. Let $\Pi$ denote the set of all causal policies. The optimal update problem for minimizing the average age can be formulated as
\begin{equation}
    \Delta^*=\underset{\pi\in\Pi}{\text{min}}\underset{n \rightarrow \infty}{\text{lim}}\frac{\mathbb{E}[\sum_{k=1}^{n}Q_k]}{\mathbb{E}[\sum_{k=1}^{n}(T_{k-1}+W_k)]}.\label{eq:problem}
\end{equation}
We assume that $(T_1,\cdots),(C_1, \cdots)$ are \emph{stationary and ergodic} random processes. By considering stationary policy, \eqref{eq:longAage} is equal to
\begin{equation}
    \overline{\Delta} =\frac{\mathbb{E}[Q_k]}{\mathbb{E}[T_{k-1}+W_k]}.\label{eq:longAageS}
\end{equation}
By substituting \eqref{def:wk} and \eqref{def:dk} into \eqref{def:qk}, we have
\begin{align*}
    \mathbb{E}[Q_k]&=\frac{1}{2}\mathbb{E}[(T_{k-1}+W_k)(T_{k-1}+W_k+2T_k+2D_k+2C_k)]\\
    &=\frac{1}{2}\mathbb{E}[(T_{k-1}+W_k)(T_{k-1}+W_k+2D_k)]\\
    &\quad +\mathbb{E}[T_{k-1}+W_k]\mathbb{E}[T_k+C_k],
\end{align*}
where the second equation holds because $W_k$ is independent with $T_k$ and $C_k$, and $T_k$'s are \emph{i.i.d.}. We can reformulate the problem \eqref{eq:problem} as
\begin{equation}
    \begin{split}
        \Delta^*=\underset{\pi\in\Pi}{\text{min}}\frac{\mathbb{E}[(T_{k-1}+W_k)(T_{k-1}+W_k+2D_k)]}{2 \mathbb{E}[T_{k-1}+W_k]}&\\
        +\mathbb{E}[T_k+C_k].& \label{eq:Problem}
    \end{split}
\end{equation}

\section{Stationary Deterministic Policies}
\label{sec:policy}

In this section, we consider \emph{stationary deterministic} policies for problem \eqref{eq:Problem}, where $W_k$ is determined based on $T_{k-1}$, $D_{k-1}$ and $C_{k-1}$. We define three functions $f,g,h$ to calculate $W_k$ depending on whether $D_{k-1}$, $C_{k-1}$ are available at time $t_k$. If $t_k<c_{k-1}$, only $T_{k-1}$ is available, we have $W_k=f(T_{k-1})<D_k$. If $c_{k-1}\le t_k<c_{k-1}+C_{k-1}$, only $C_{k-1}$ is not available, we have $D_{k-1}\le W_k=g(T_{k-1},D_{k-1})<D_{k-1}+C_{k-1}$. If $t_k\ge c_{k-1}+C_{k-1}$, then $T_{k-1}$, $D_{k-1}$ and $C_{k-1}$ are all available, we have $W_k=h(T_{k-1},D_{k-1},C_{k-1})\ge D_{k-1}+C_{k-1}$. In this paper, we consider three policies. In the \emph{long wait} policy, the source node only sends a new packet after the old one is delivered. In the \emph{peak age threshold} policy, the source node sends a new packet after the old one starts to be processed. In the \emph{peak age threshold with postponed plan} policy, we consider reducing the waiting time of the packet in the buffer by a postponed plan. These policies are detailed as follows.
\subsection{Long Wait Policy}
Consider the source node can only send a new packet $k$ after packet $k-1$ is delivered, i.e.,
\begin{equation}
    t_k\ge c_{k-1}+C_{k-1}.\label{cons:longw}
\end{equation}
After packet $k$ is sent, the source node will wait until it is delivered to the destination. In this case, packet will be processed immediately when it arrives at the server. Thus, we have $D_k=0$ and $W_k\ge C_{k-1}$. A policy $\pi\in \Pi_{SD}$ is said to be a \emph{long wait} policy if $W_k=h(T_{k-1},C_{k-1})$ and $f(\cdot )=g(\cdot )=\infty$.

Since $T_k$ and $C_k$ are \emph{i.i.d.}, the status update system in this case is equivalent to a system composed of a sources node, a delay channel and a destination, as studied in \cite{SunUp}. The transmission plus processing time $T_k+C_k$ in this paper is equivalent to the transmission time $Y_k$ in \cite{SunUp}. The interval between the delivery of packet $k-1$ and the generation of packet $k$, termed as $W_k-C_{k-1}$, is equivalent to $Z_k$ in \cite{SunUp}. Therefore, we have the optimal long wait policy based on \cite[Theorem 4]{SunUp} as follows.
\begin{lemma}\label{newLemma1}
    If $\mathbb{E}[T]+\mathbb{E}[C]\ge 0$, the optimal long wait policy to problem \eqref{eq:problem} is
    \begin{equation}
        h(T_{k-1},C_{k-1})=\max\{C_{k-1},\beta-T_{k-1}\},\label{eq:longwh}
    \end{equation}
    where $\beta$ satisfies
    \begin{equation}
        \mathbb{E}[T+C+h(T,C)]=\frac{\mathbb{E}[\left(T+C+h(T,C)\right)^2]}{2\beta}.\label{eq:optimalbeta}
    \end{equation}
\end{lemma}
The parameter $\beta$ can be calculated by Algorithm 2 in \cite{SunUp}. The optimal long wait policy has a threshold structure. The source node tends to send packet $k$ at time $t_{k-1}+\beta$ if the constraint \eqref{cons:longw} is satisfied, otherwise, it will send packet $k$ at $t_{k-1}+T_{k-1}+C_{k-1}$ as $\beta<T_{k-1}+C_{k-1}$. The transmission interval of adjacent packets, $t_k-t_{k-1}$, has a threshold structure
\begin{equation*}
    t_{k}-t_{k-1}=\max\{\beta,T_{k-1}+C_{k-1}\}.
\end{equation*}

There is an example of the source node transmitting a sequence of packets depicted in Fig. \ref{fig:sub}, where the red box represents the transmission time and the blue box represents the processing time. As shown in Fig. \ref{fig:sub} (a), under the long wait policy, the start of the red box is always behind the end of the previous blue box. An example of a non-long wait policy is shown in Fig. \ref{fig:sub} (b). Compared with the long wait policy, the source node sends new packets when the server is busy. Intuitively, the non-long wait policy fully utilizes the transmission and processing ability of the system. It is expected to improve AoI performance, which however, requires carefully design as it may lead to unexpected waiting in the buffer. Our goal in this paper is to design better policies to reduce the average AoI.

\begin{figure}
    \centering
    \subfloat[]{\includegraphics[width=0.7\linewidth]{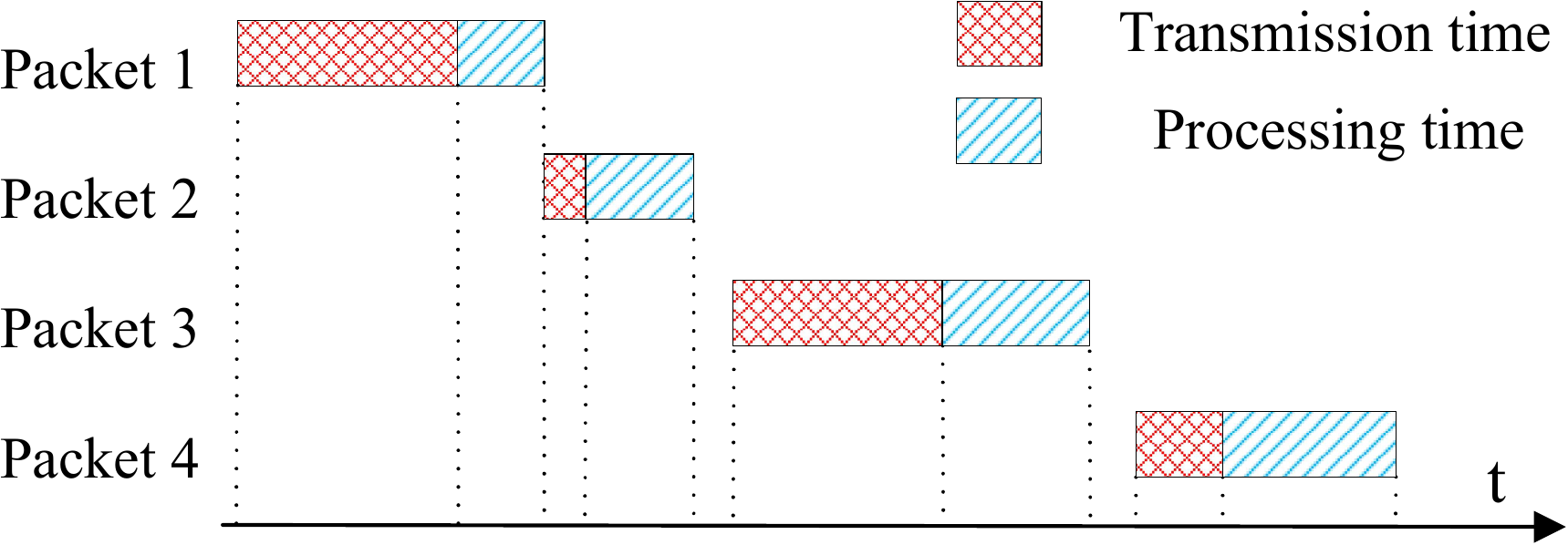}}

    \subfloat[]{\includegraphics[width=0.7\linewidth]{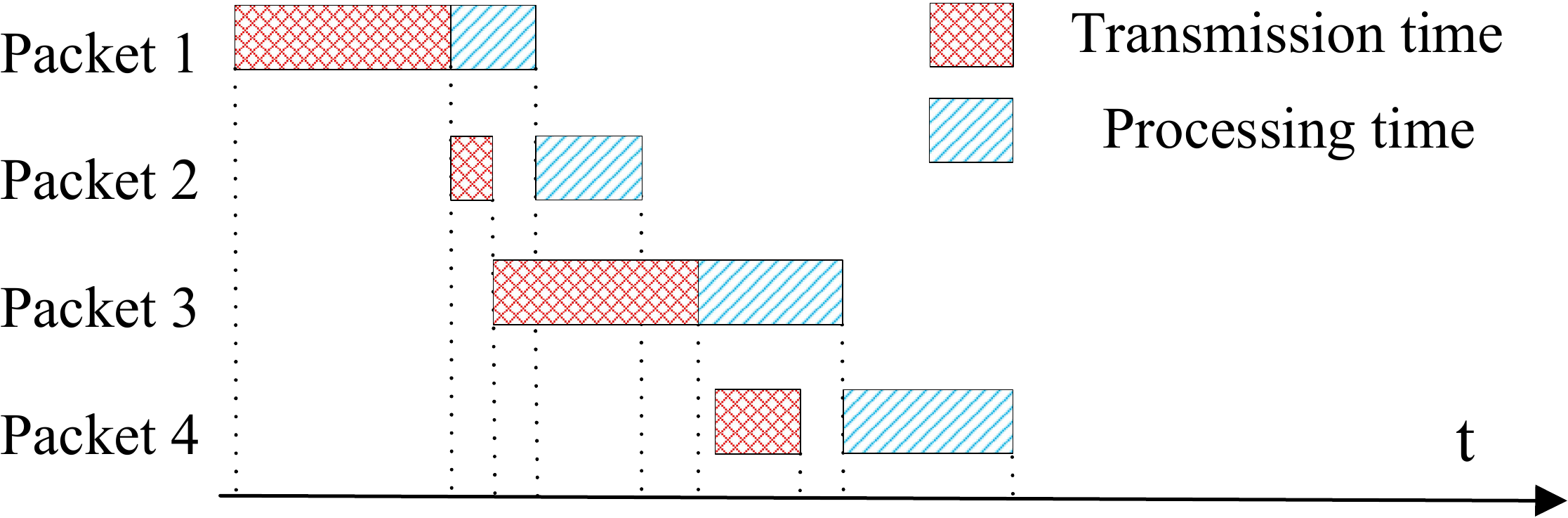}}
    \caption{Example of the source node transmits packet (a)with constraints \eqref{cons:longw} and (b)without constraints \eqref{cons:longw}.}
    \label{fig:sub}
\end{figure}

\subsection{Peak Age Threshold Policy}

Consider the source node can send a new packet $k$ after packet $k-1$ starts to be processed, i.e.,
\begin{equation}
    t_k\ge c_{k-1}.\label{cons:loosercons}
\end{equation}
In this case, there is at most one packet in the buffer, that is, the packet arrives at the server while the previous packets are still under processing.

We consider the \emph{peak age of information} (PAoI) proposed in \cite{PAoI}. The PAoI of packet $k$ is the height of the highest point of the area $Q_k$ in \eqref{fig:age}, i.e., $A_k=c_k+C_k-t_{k-1}$. Suppose the source node sends packet $k$ at time $t\ge c_{k-1}$, the PAoI $A_k$ can be estimated according to the distribution of $T_k,C_{k-1}$.
In particular, if $t\ge c_{k-1}+C_{k-1}$, $c_k$ can be estimated by $\hat{c}_{k}=t_k+\mathbb{E}[T_k]$. Let $\hat{A}_k^h$ denote the estimation of $A_k$ in this case, we have
\begin{equation}
    \hat{A}_k^h=t+\mathbb{E}[T_k]+\mathbb{E}[C_k]-t_{k-1}.\label{eq:akh}
\end{equation}
If $c_{k-1}\le t<c_{k-1}+C_{k-1}$, we have $\hat{c}_k=\max\{t+\mathbb{E}[T_{k}],c_{k-1}+\mathbb{E}[C_{k-1}|C_{k-1}>t-c_{k-1}]\}$ based on \eqref{eq:ck}, in which $\mathbb{E}[C_{k-1}|C_{k-1}>t]$ is a conditional probability and is non-decreasing with respect to $t$ for any distribution of $C_{k-1}$. Let $\hat{A}_k^g$ denote the estimation of $A_k$ in this case, we have
\begin{equation}
    \begin{split}
        \hat{A}_k^g=&\max\{t+\mathbb{E}[T_{k}],c_{k-1}+\mathbb{E}[C_{k-1}|C_{k-1}>t-c_{k-1}]\}\\
                    &+\mathbb{E}[C_k]-t_{k-1}.\label{eq:akg}
    \end{split}
\end{equation}
Note that $\hat{A}_k^h$ and $\hat{A}_k^g$ are non-decreasing with respect to $t$. Based on the long wait policy, we propose a heuristic, peak age threshold-based policy. With constraint \eqref{cons:loosercons}, we need to set $f(\cdot)=\infty$. A policy $\pi$ is said to be a \emph{peak age threshold} policy if
\begin{align}
    h_p(T_{k-1},D_{k-1},C_{k-1})&=\widetilde{t}^h_p-t_{k-1}-T_{k-1},\\
    g_p(T_{k-1},D_{k-1})&=\widetilde{t}^g_p-t_{k-1}-T_{k-1},
\end{align}
where
\begin{align}
    \widetilde{t}^g_p&=\underset{t\ge c_{k-1},\hat{A}_k^g\ge \lambda}{\text{arg min}}\hat{A}_k^g,\label{def:tg}\\
    \widetilde{t}^h_p&=\underset{t\ge c_{k-1}+C_{k-1}, \hat{A}_k^h\ge \lambda}{\text{arg min}}\hat{A}_k^h.\label{def:th}
\end{align}
The peak age threshold policy is an online policy. After packet $k-1$ starts to be processed, the source node plans to send packet $k$ at time $\widetilde{t}^g_p$ calculated from \eqref{def:tg} unless packet $k-1$ is delivered before time $\widetilde{t}^g_p$. Specifically, if at time $t<\widetilde{t}^g_p$, the packet $k-1$ is delivered, then the source node will calculate $\widetilde{t}^h_p$ from \eqref{def:th} and send packet $k$ at this time. From \eqref{eq:akg} we have $\hat{A}_k^h\ge \mathbb{E}[T]+\mathbb{E}[C]$. Thus, when $\lambda\le \mathbb{E}[T]+\mathbb{E}[C]$, the source node will immediately send a new packet when the old packet starts to be processed as $\widetilde{t}^g=c_{k-1}$ for all $k$.

We consider using a line search algorithm to numerically compute the optimal $\lambda$. The following lemma shows that when $\lambda$ is sufficiently large, the average AoI $\Delta_{p}$ under the peak age threshold policy increases linearly with respect to $\lambda$.

\begin{lemma}\label{newLemma2}
    When $\lambda$ is sufficiently large, we have
    \begin{equation}
        \Delta_{p}=\frac{1}{2}\left(\lambda+\mathbb{E}[T]+\mathbb{E}[C]\right),\label{eq:bigD}
    \end{equation}
    with probability 1.
\end{lemma}
\begin{proof}
When $\lambda$ is sufficiently large, the source node will wait a long time after packet $k-1$ is sent before sending packet $k$. Thus, $W_k$ is sufficiently large so that $P[W_k>T_{k-1}+C_{k-1}]=1$ and $P[D_k=0]=1$ for any distribution of $T_{k-1},C_{k-1}$. Therefore, the source will send packet $k$ at time $t=\widetilde{t}^h=\lambda+t_{k-1}-\mathbb{E}[T]-\mathbb{E}[C]$ with probability 1, that is, $W_k=\lambda-T_{k-1}-\mathbb{E}[T]-\mathbb{E}[C]$. By substituting $W_k$ and $P[D_k=0]=1$ into \eqref{eq:Problem}, we have \eqref{eq:bigD}.
\end{proof}
Based on Lemma \ref{newLemma2}, we can find the optimal $\lambda$ in the interval $[\mathbb{E}[T]+\mathbb{E}[C],u]$ with arbitrary precision through an exhaustive search, where $u$ is a sufficiently large positive number. Numerical simulations in Sec. \ref{sec:numerical} show that the optimal search can be made with relatively small $u$.

\subsection{Peak Age Threshold Policy With Postponed Plan}

Under constraint \eqref{cons:loosercons}, the waiting time for a packet in the buffer is destructive for AoI minimization. We expect that a packet can be processed immediately when it arrives at the server. In the peak age threshold policy, we can achieve this by increasing $\lambda$, however, this may also increase age. Therefore, we try to postpone the packet transmission plan when the server is still busy based on the peak threshold age policy: 1) If the source node plans to send packet $k$ at time $\widetilde{t}^g_p<c_{k-1}+C_{k-1}$ before packet $k-1$ is delivered, the source defers the plan until time $t$, which satisfies $t+\mathbb{E}[T]\ge c_{k-1}+\mathbb{E}[C_{k-1}|C_{k-1}>t-c_{k-1}]$, or packet $k-1$ is delivered. 2) If the source node plans to send packet $k$ after packet $k-1$ is delivered, it sends at time $\widetilde{t}^h_p$ the same as the peak age threshold policy. Let $\Omega$ denote the set of $l\ge 0$ satisfying $l+\mathbb{E}[T]\ge \mathbb{E}[C|C>l]$. The peak age threshold policy with postponed plan satisfies $g_{pp}(T_{k-1},D_{k-1})=g_p(T_{k-1},D_{k-1})$ and
\begin{equation}
    h_{pp}(T_{k-1},D_{k-1},C_{k-1})=\widetilde{t}^h_{pp}-t_{k-1}-T_{k-1},
\end{equation}
where
\begin{equation}
    \widetilde{t}^g_{pp}=\underset{t-c_{k-1}\in \Omega,\hat{A}_k^g\ge \lambda}{\text{arg min}}\hat{A}_k^g.\label{eq:tpo}
\end{equation}

When $C_k$'s obey the exponential distribution, based on the memoryless property, we have $\mathbb{E}[C|C>l]=\mathbb{E}[C]+l$ for $l\ge 0$. In this special case, if $\mathbb{E}[T]\ge \mathbb{E}[C]$, we have $\Omega=\{l|l\ge 0\}$, thus the postponed plan has no effect on the peak age threshold policy as $\widetilde{t}^g_{pp}=\widetilde{t}^g_p$. If $\mathbb{E}[T]< \mathbb{E}[C]$, $\Omega=\emptyset$, thus the equation \eqref{eq:tpo} has no solution. The source will not send new packet until the previous packet is delivered. Therefore, the peak age threshold policy with postponed plan degrades to the long wait policy and the optimal $\lambda$ is $\beta +\mathbb{E}[T]+\mathbb{E}[C]$.

\section{Numerical results}
\label{sec:numerical}
\begin{figure}[!t]
    \centering
    \includegraphics[width=0.8\linewidth]{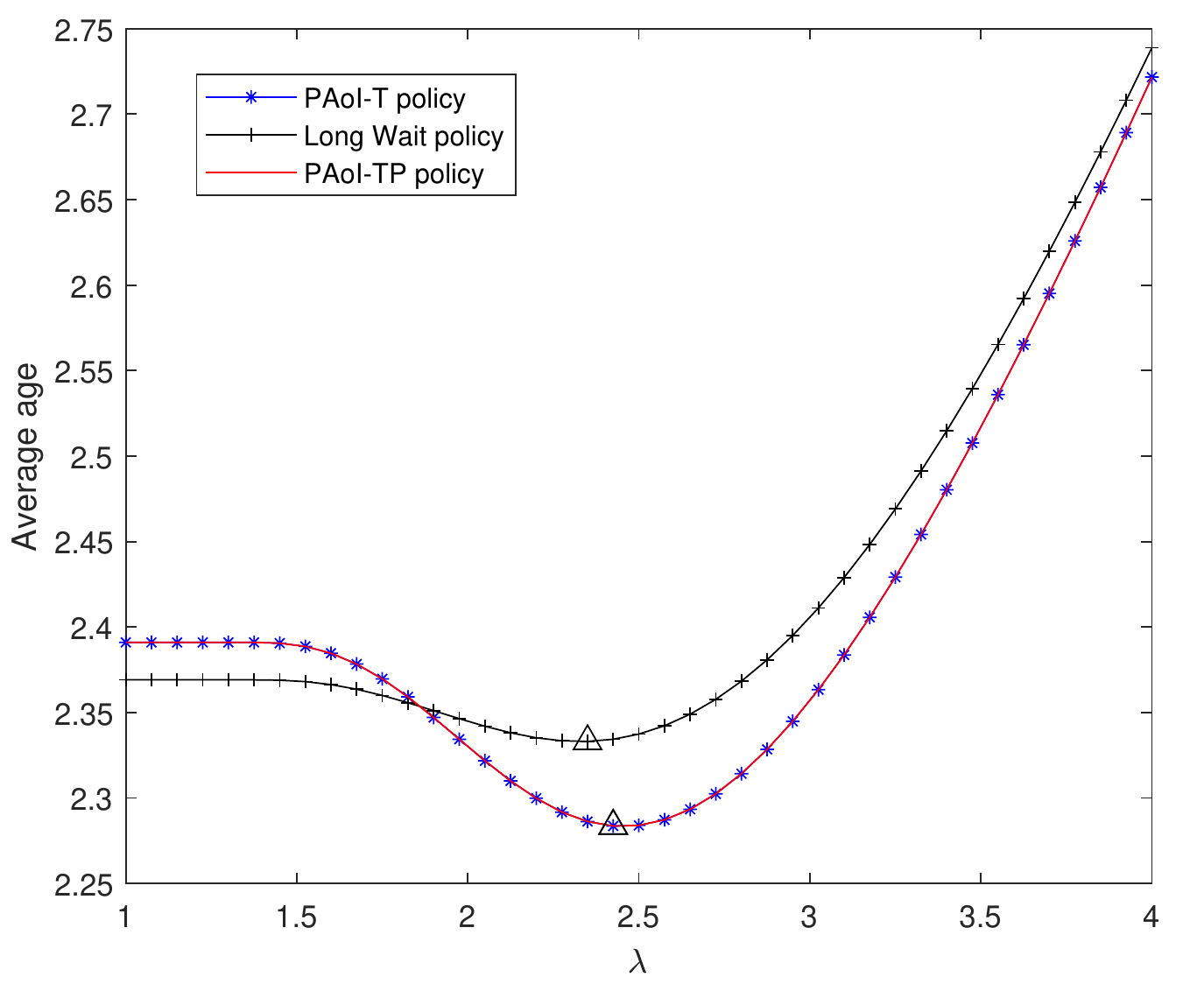}
    \caption{Average age for \emph{i.i.d.} exponentially distributed transmission and processing time with $\mathbb{E}[T]/\mathbb{E}[C]=4$.}\label{fig:E1}
\end{figure}
\begin{figure}[htp]
    \centering
    \includegraphics[width=0.8\linewidth]{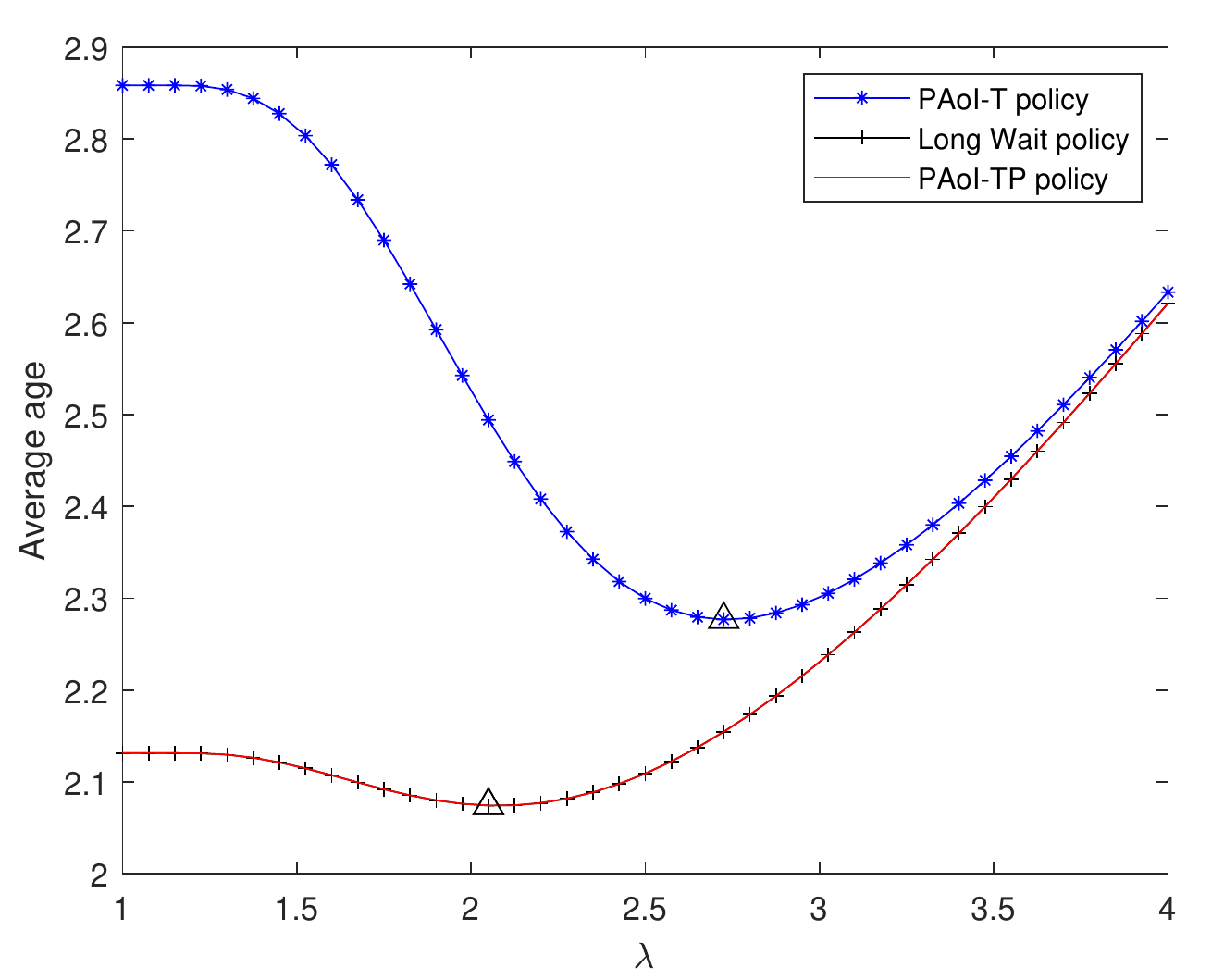}
    \caption{Average age for \emph{i.i.d.} exponentially distributed transmission and processing time with $\mathbb{E}[T]/\mathbb{E}[C]=0.25$.}\label{fig:E2}
\end{figure}

\begin{figure}[!t]
    \centering
    \includegraphics[width=0.8\linewidth]{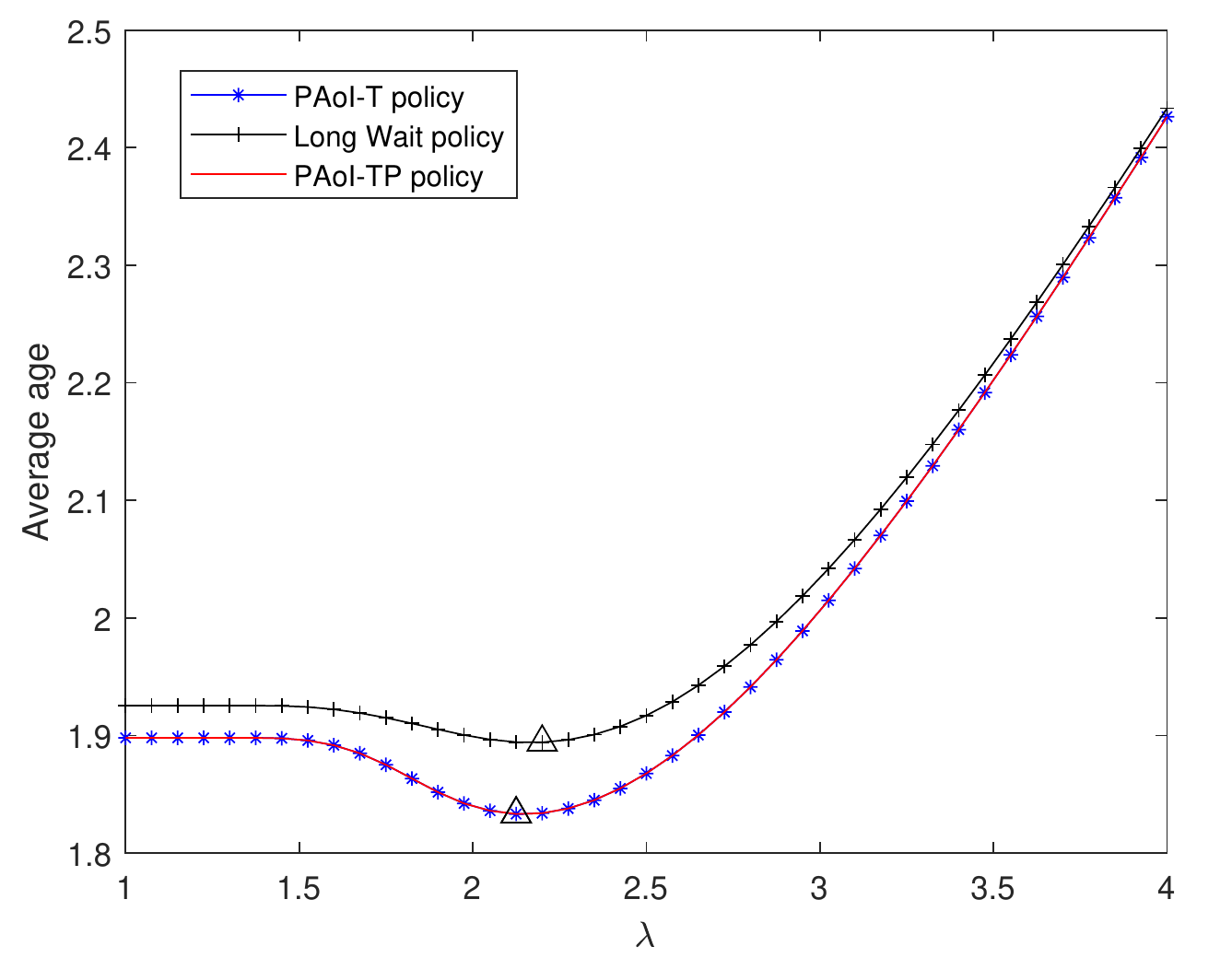}
    \caption{Average age for \emph{i.i.d.} exponentially distributed transmission time and uniformly distributed processing time with $\mathbb{E}[T]/\mathbb{E}[C]=4$.}\label{fig:E3}
\end{figure}
\begin{figure}[!t]
    \centering
    \includegraphics[width=0.8\linewidth]{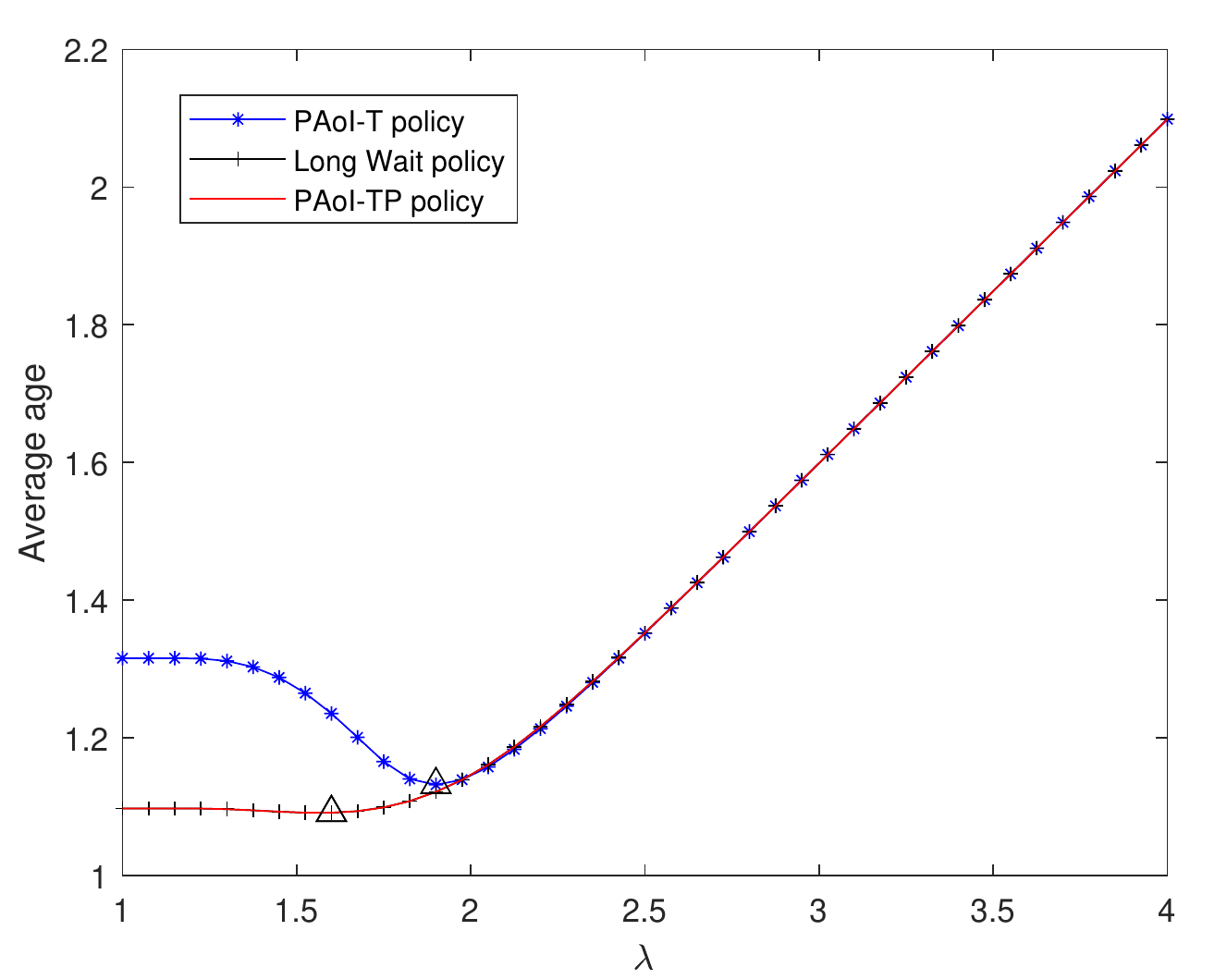}
    \caption{Average age for \emph{i.i.d.} exponentially distributed transmission time and uniformly distributed processing time with $\mathbb{E}[T]/\mathbb{E}[C]=0.25$.}\label{fig:E4}
\end{figure}
\begin{figure}[!t]
    \centering
    \includegraphics[width=0.8\linewidth]{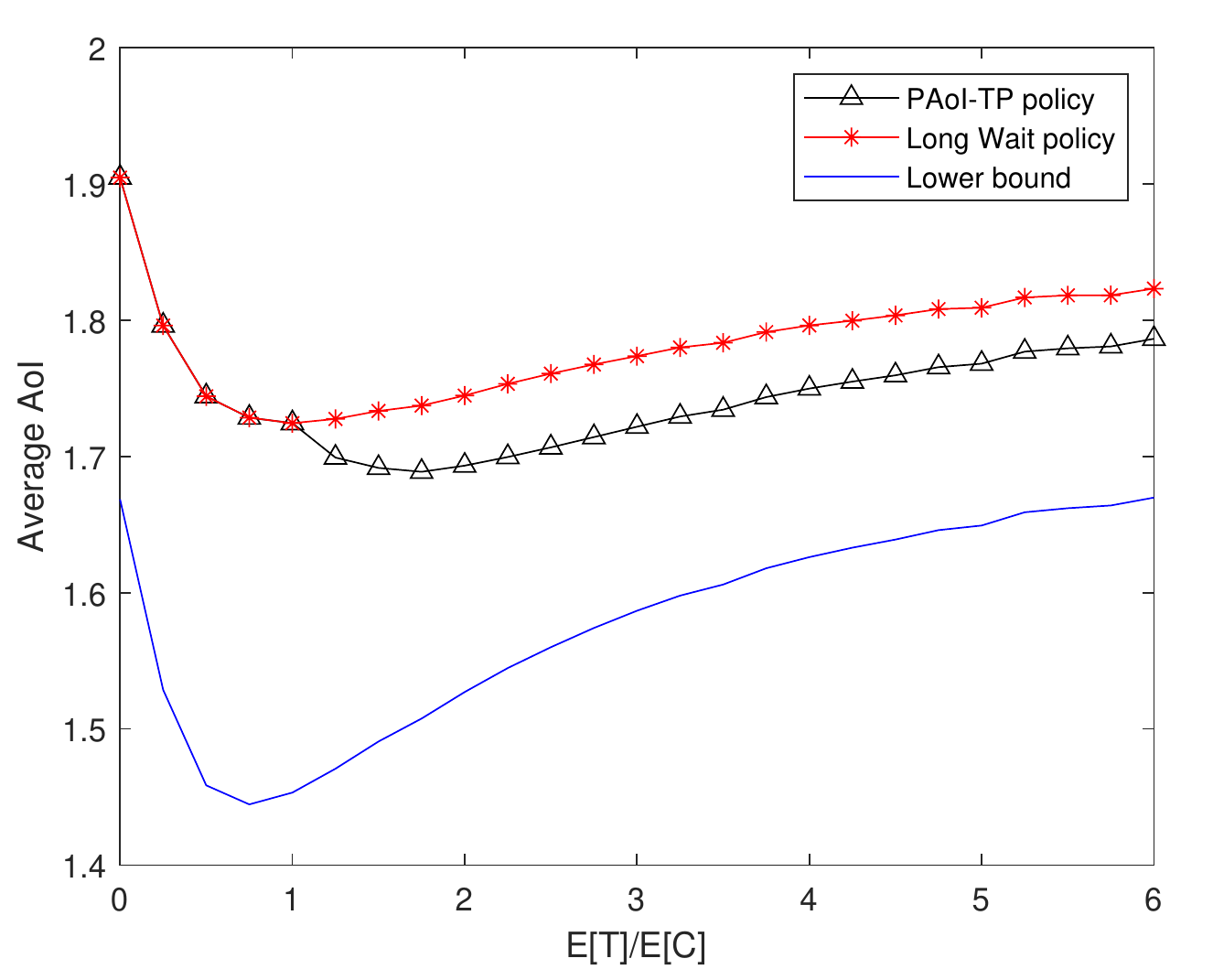}
    \caption{Average age for \emph{i.i.d.} exponentially distributed transmission time and processing time.}\label{fig:E5}
\end{figure}

In this section, we study the average age through numerical simulation. $\mathbb{E}[T]+\mathbb{E}[C]$ remains unchanged at 1 for all experiments and $\mathbb{E}[T]/\mathbb{E}[C]$ changes for different experiments. The initial value $T_0$ is set to 1 and $C_0$ is set to 0. We simulate 100,000 packets to numerically calculate the average age. Firstly, we analyze the average age with exponentially distributed transmission time and processing time. The average age as a function of the threshold $\lambda$ is shown in Fig. \ref{fig:E1} under the peak age threshold policy (PAoI-T) and the peak age threshold policy with postponed plan (PAoI-TP). The long wait policy is also shown in Fig. \ref{fig:E1}, under which the parameter $\lambda$ refers to $\beta+\mathbb{E}[T]+\mathbb{E}[C]$. As shown before, the PAoI-T policy and the PAoI-TP policy are the same in this case. With the optimal parameter $\lambda = 2.425$, the two policies have a lower average age than the long wait policy with optimal parameter $\lambda = 2.35$. This indicates that the source node sending a new packet before the previous one is delivered can effectively reduce the average age. With $\mathbb{E}[T]/\mathbb{E}[C]=0.25$, the average age result is shown in Fig. \ref{fig:E2}. In this case, the PAoI-TP policy is equivalent to the long wait policy. It can be seen that the optimal long wait policy has lower average age than the PAoI-T policy. When $\lambda<2$, the average age under the PAoI-T policy is much higher than that under the long wait policy. This is because $\mathbb{E}[T]$ is much smaller than $\mathbb{E}[C]$. In this case, the probability of packets entering the buffer is high when $\lambda$ is small.

In Figs. \ref{fig:E3} and \ref{fig:E4}, the average age with exponentially distributed transmission time and uniformly distributed processing time is shown with $\mathbb{E}[T]/\mathbb{E}[C]=4$ and $\mathbb{E}[T]/\mathbb{E}[C]=0.25$, respectively. The PAoI-T policy is equivalent to the PAoI-TP policy in Fig. \ref{fig:E3} and the long wait policy in Fig. \ref{fig:E4} similar to the case of exponentially distributed processing time. The optimal PAoI-T policy has a smaller average age in both cases.

Based on Lemma \ref{newLemma2}, the average age for the problem \eqref{eq:Problem} under the PAoI-T policy increases linearly with $\lambda$ when $\lambda$ is sufficiently large. The above experiments show that $[1,4]$ is a sufficient interval to search for the optimal $\lambda$ when $\mathbb{E}[T]+\mathbb{E}[C]=1$. Therefore, the average age for the optimal PAoI-TP policy for $\mathbb{E}[C]/\mathbb{E}[C]$ varies from 0 to 6 is shown in Fig. \ref{fig:E5} with exponentially distributed transmission time and processing time. The optimal $\lambda$'s are obtained by searching in the interval $[1,4]$. The optimal long wait policy from Lemma \ref{newLemma1} is also shown in Fig. \ref{fig:E5}. The lower bound in the figure is obtained by solving the offline version of the problem \eqref{eq:Problem} under constant \eqref{cons:loosercons} using the convex optimization method. When $\mathbb{E}[T]/\mathbb{E}[C]>1$, the PAoI-TP policy is better than the long wait policy.

\section{conclusion}
\label{sec:conclusion}

In this paper, we studied how to jointly schedule transmission and processing in a status-update system with mobile edge computing to optimize the average AoI. The source node only sends a new packet after the old packet is delivered under the long wait policy. Under the peak age threshold policy, the source node will estimate the peak AoI of the packet to be transmitted and only transmit it when the estimation exceeds a threshold, regardless of the busy or idle state of the edge server. An improved policy, called the peak age threshold policy with postponed plan, is proposed based on the consideration of reducing the waiting time of a packet in the buffer. When the processing time is exponentially distributed, the age behavior is strongly dependent on $\mathbb{E}[T]/\mathbb{E}[C]$. When $\mathbb{E}[T]/\mathbb{E}[C]<1$, the peak age threshold policy with postponed plan degrades to a long wait policy. In this case, the postponed plan effectively reduces the AoI. When $\mathbb{E}[T]/\mathbb{E}[C]>1$, the postponed plan has no effect on the peak age threshold policy. In this case, the peak age threshold policy is better than the long wait policy.

\end{document}